\documentclass[aps,pra,reprint,showpacs,floatfix,longbibliography]{revtex4-1}
\usepackage{graphicx, amsmath, amssymb, amsthm, braket, enumerate}
\usepackage[caption=false]{subfig}
\allowdisplaybreaks[1] % http://tex.stackexchange.com/questions/57338/how-can-i-allow-gather-to-break-between-pages

\newtheorem{definition}{Definition}
\newtheorem{lemma}{Lemma}
\newtheorem{theorem}{Theorem}

\newcommand{\ketbra}[2]{{\left| #1 \middle\rangle \middle\langle #2 \right|}}
\newcommand{\Mod}[1]{\ (\text{mod}\ #1)}

\begin{document}

\title{Stationary States in Quantum Walk Search}

\author{Kri\v{s}j\={a}nis Pr\={u}sis}
	\email{krisjanis.prusis@lu.lv}
\author{Jevg\={e}nijs Vihrovs}
	\email{jevgenijs.vihrovs@lu.lv}
\author{Thomas G.~Wong}
	\email{Currently at University of Texas at Austin, twong@cs.utexas.edu}
	\affiliation{Faculty of Computing, University of Latvia, Rai\c{n}a bulv.~19, R\=\i ga, LV-1586, Latvia}

\begin{abstract}
	When classically searching a database, having additional correct answers makes the search easier. For a discrete-time quantum walk searching a graph for a marked vertex, however, additional marked vertices can make the search harder by causing the system to approximately begin in a stationary state, so the system fails to evolve. In this paper, we completely characterize the stationary states, or 1-eigenvectors, of the quantum walk search operator for general graphs and configurations of marked vertices by decomposing their amplitudes into uniform and flip states. This infinitely expands the number of known stationary states and gives an optimization procedure to find the stationary state closest to the initial uniform state of the walk. We further prove theorems on the existence of stationary states, with them conditionally existing if the marked vertices form a bipartite connected component and always existing if non-bipartite. These results utilize the standard oracle in Grover's algorithm, but we show that a different type of oracle prevents stationary states from interfering with the search algorithm.
\end{abstract}


\maketitle

%-------------------------------------------------------------------------------
% Main Matter
%-------------------------------------------------------------------------------

\section{Introduction}

Quantum computers are well-known for their ability to outperform classical computers in many algorithmic applications \cite{Montanaro2016}. One famous example is unstructured search, where an item in a database is marked by an oracle that responds yes or no as to whether an item is marked. A classical computer finds the marked item in $O(N)$ queries, while a quantum computer takes $O(\sqrt{N})$ queries using Grover's algorithm \cite{Grover1996}. If there are $k$ marked items, then the classical and quantum computers respectively search in $O(N/k)$ and $O(\sqrt{N/k})$ time. So additional marked items make the search problem easier for both types of computers, as expected.

If there is structure to the database, however, then there are scenarios where additional marked items make search easier for a classical computer but harder for a quantum computer. More precisely, say the database is formulated as a graph of $N$ vertices where one or more vertices are marked. The edges of the graph define the structure by which one moves from vertex to vertex. To find a marked vertex, one approach is to classically and randomly walk on the graph, querying the oracle with each step until a marked vertex is found. Then the more marked vertices, the easier the search problem becomes since there will be more marked vertices for the random walk to stumble upon.

The opposite can be true in the quantum regime, where additional marked vertices make the problem harder, not easier. This occurs using a discrete-time quantum walk \cite{Kempe2003}, where the vertices of the graph define an $N$-dimensional Hilbert space. Furthermore, to effect a non-trivial evolution \cite{Meyer1996a,Meyer1996b}, the walking particle also has an internal degree of freedom (often called the coin) encoding the various directions that the particle can hop. Then if the number of edges of the graph is $|E|$, the system evolves in the Hilbert space $\mathbb{C}^{2|E|}$ spanned by orthonormal basis states $\{ \ket{v,c} \}$, where $v = 1, 2, \dots, N$ specifies the vertex and $c = 1, 2, \dots, d_v$ denotes the  directions in which a particle at vertex $v$ can point (so $d_v$ denotes the degree of vertex $v$). The system $\ket{\psi}$ begins in a uniform state
\begin{equation}
	\label{eq:initial}
	\ket{\psi(0)} = \frac{1}{\sqrt{2|E|}} \sum_{v=1}^N \sum_{c=1}^{d_v} \ket{v,c}.
\end{equation}
Then the search is performed by repeated applications of
\begin{equation}
	\label{eq:U}
	U = SCQ.
\end{equation}
Here, $S$ is the flip-flop shift \cite{AKR2005} that causes the particle to hop and then turn around (so a particle at vertex $1$ pointing towards vertex $2$ will end up at vertex $2$ pointing towards vertex $1$, or $S\ket{1,2} = \ket{2,1}$). $C$ is the coin flip that applies the Grover diffusion coin $2 \ketbra{s_c}{s_c} - I$ \cite{SKW2003} on the directional space at each vertex, where $\ket{s_c} = \sum_{i=1}^{d_v} \ket{i}/\sqrt{d_v}$ is the equal superposition over the directions. As shown in Lemma~2 of \cite{APVW2016}, the Grover coin inverts each amplitude of the directional state about the average amplitude, so it is the ``inversion about the mean'' of Grover's algorithm \cite{Grover1996}. Finally, $Q$ is an oracle query that flips the sign of the amplitude at marked vertices, which is the standard oracle in Grover's algorithm.

As shown by Ambainis, Kempe, and Rivosh \cite{AKR2005}, applying this quantum algorithm to search for a unique marked vertex on the two-dimensional (2D) periodic square lattice yields a success probability of $O(1/\log N)$ after $O(\sqrt{N \log N})$ steps. With amplitude amplification \cite{BHMT2000}, this results in an overall runtime of $O(\sqrt{N} \log N)$. Now say there are two marked vertices that are adjacent to each other. Classically, this makes the search problem easier. But Nahimovs and Rivosh \cite{NR2016} recently showed that the quantum walk search algorithm now takes time $O(N)$, completely losing its quantum speedup. This is because the initial uniform state \eqref{eq:initial} is approximately equal to a 1-eigenvector of the search operator $U$ \eqref{eq:U}, and so the system fails to evolve, and the quantum algorithm is equivalent to classically guessing and checking. So in this example, having an additional marked vertex makes the search problem harder, not easier, for the quantum algorithm.

\begin{figure}
\begin{center}
	\includegraphics{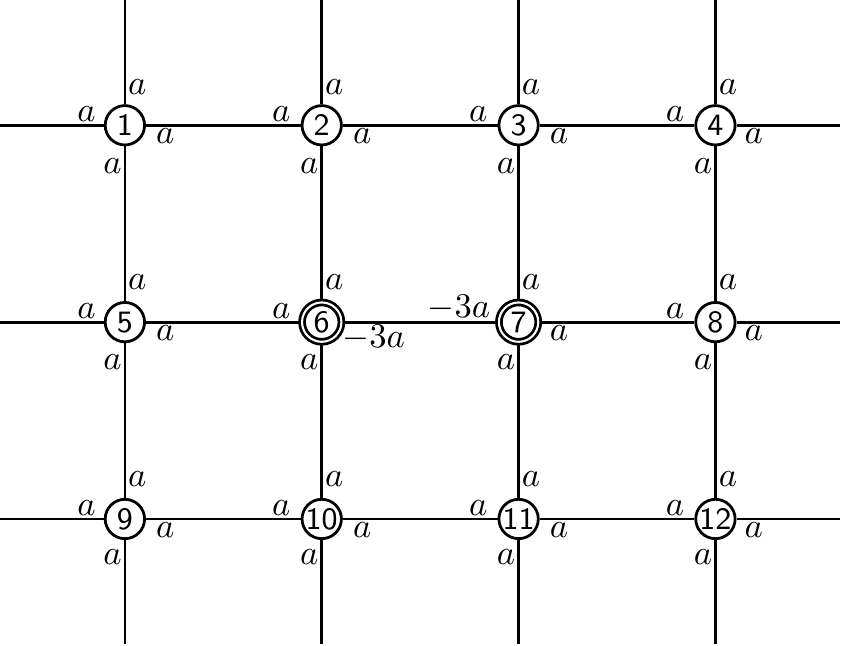}
	\caption{\label{fig:square_pair_opt} The optimal stationary state for the $4 \times 3$ periodic square grid with an adjacent pair of marked vertices (vertices $6$ and $7$, indicated by double circles).}
\end{center}
\end{figure}

More precisely, the stationary state that approximates the initial uniform state \eqref{eq:initial} is depicted in Fig.~\ref{fig:square_pair_opt}. This figure reveals three properties, identified by Nahimovs and Rivosh \cite{NR2016} and further explored in follow-up work by Nahimovs and Santos \cite{NS2016}, of stationary states. First, the directional amplitudes of unmarked vertices (depicted by single circles) are equal. For example, vertex $1$ has amplitude $a$ (chosen to normalize the overall state) in each of its four directions. Second, the directional amplitudes of marked vertices (depicted by double circles) sum to $0$. For example, the sum of vertex $6$'s amplitudes is $a - 3a + a + a = 0$. Third, the directional amplitudes of adjacent vertices pointing to each other are equal. For example, vertices $6$ and $7$ point to each other with amplitude $-3a$. Nahimovs, Rivosh, and Santos \cite{NR2016,NS2016} showed that a state with these three properties is a stationary state (i.e., 1-eigenvector) of the quantum walk search operator $U$ \eqref{eq:U}. Going through each operator in $U$, the query $Q$ flips the sign of the marked vertices, the coin $C$ again flips the sign of the marked vertices since their average amplitudes are zero, and the shift $S$ swaps pairs of amplitudes pointing to each other, which are equal. Thus $U$ leaves such states invariant.

\begin{figure}
\begin{center}
	\includegraphics{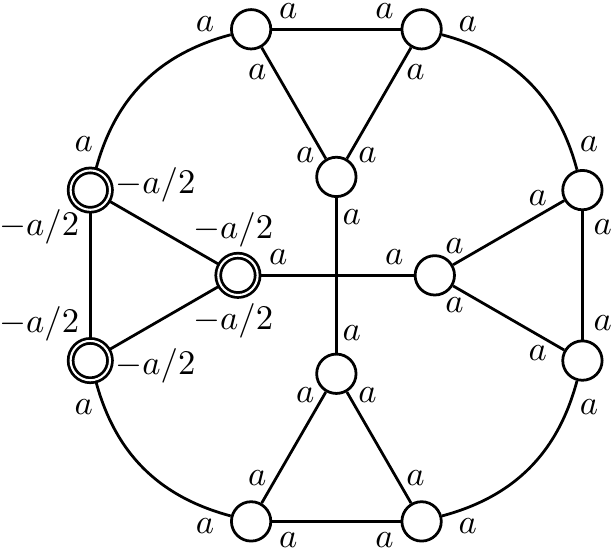}
	\caption{\label{fig:simplex_full} The optimal stationary state for the simplex of $4$ complete graphs, each with $3$ vertices, with a completely marked clique (indicated by double circles).}
\end{center}
\end{figure}

Another example, which to the best of our knowledge is historically the first example of the quantum walk search algorithm beginning in an approximate stationary state, is the simplex of complete graphs with a fully marked clique \cite{WA2015}. An example is depicted in Fig.~\ref{fig:simplex_full}, and the labels express the stationary state satisfying the three conditions of Nahimovs, Rivosh, and Santos \cite{NR2016,NS2016}. Since the initial uniform state \eqref{eq:initial} is approximately equal to this stationary state, the quantum algorithm is no better than classically guessing and checking.

In this paper, however, we show that the stationary states described by Nahimovs, Rivosh, and Santos \cite{NR2016,NS2016} are not the only stationary states. We give the exact necessary and sufficient conditions on the amplitudes between adjacent vertices for a state to be stationary. In doing so, we greatly expand the number of known stationary configurations to an infinite class. We do this by incorporating concepts from \cite{APVW2016} regarding uniform and flip states, which form an orthogonal basis for directional states. Then we show that the type of stationary state described by Nahimovs, Rivosh, and Santos \cite{NR2016,NS2016} is optimal, in the sense that it is the stationary states closest to the initial uniform state \eqref{eq:initial}. Then we prove two theorems about the existence of stationary states on general graphs and with a connected component of marked vertices. In particular, if the marked vertices form a bipartite component, then a stationary state exists if and only if the sum of the amplitudes to be assigned to each partite set are equal. If they form a non-bipartite component, on the other hand, then a stationary state always exists.

These results use the search operator $U$ \eqref{eq:U}, which utilizes the natural oracle from Grover's algorithm. If we instead use the oracle from \cite{SKW2003}, which applies different coin operators on unmarked and marked vertices, then stationary states are always orthogonal to the initial uniform state \eqref{eq:initial}. Hence the most natural search oracle may not be ideal for certain quantum walk search problems.

%-------------------------------------------------------------------------------
% Section
%-------------------------------------------------------------------------------

\section{Stationary States}

In this section, we begin by introducing the concepts of uniform and flip states from \cite{APVW2016}, which form an orthogonal basis for directional states. Then we give conditions for a state to be stationary under the quantum walk search operator $U$ \eqref{eq:U}. This exactly characterizes the stationary states of the walk, and it greatly expands stationary states beyond those of Nahimovs, Rivosh, and Santos \cite{NR2016,NS2016}. Nonetheless, we prove that their stationary states are optimal, meaning they are the stationary states closest to the initial uniform state \eqref{eq:initial}.

%-------------------------------------------------------------------------------

\subsection{Uniform and Flip States}

Let $\ket{v_c}$ be the (likely unnormalized) directional state at vertex $v$. For example, in Fig.~\ref{fig:square_pair_opt}, the directional state at vertex $6$ is $\ket{6_c} = a \ket{\uparrow} - 3a \ket{\rightarrow} + a \ket{\downarrow} + a \ket{\leftarrow}$. In general, we can express a directional state as $\ket{v_c} = \sum_{i=1}^{d_v} \alpha_i \ket{i}$. From this, we define uniform and flip states:
\begin{definition}
	We call $\ket{v_c}$ a \emph{uniform state} if $\alpha_1 = \alpha_2 = \ldots = \alpha_{d_v}$.
\end{definition}
\begin{definition}
	We call $\ket{v_c}$ a \emph{flip state} if $\sum_{i=1}^{d_v} \alpha_i = 0$.
\end{definition}
Note that \cite{APVW2016} defined uniform and flip states with regard to all vertices, whereas we only define them here with regards to individual vertices.

Uniform and flip states are useful because they are a complete orthogonal basis for directional states \cite{APVW2016}. To prove this, we first show that any uniform state $\ket{\sigma}$ is orthogonal to any flip state $\ket{\phi}$:
\[ \braket{\phi | \sigma} = \sum_{i=1}^{d_v} \braket{\phi | i} \braket{i | \sigma} = \overline{\sigma} \sum_{i=1}^{d_v} \braket{\phi | i} = 0, \]
where $\overline{\sigma} = \frac{1}{d_v} \sum_{i=1}^{d_v} \braket{i|\sigma}$ denotes the average of the amplitudes of $\ket{\sigma}$.

Now let us show that uniform and flip states are a complete basis. Consider an arbitrary directional state $\ket{v_c}$, and let $\overline{v_c} = \frac{1}{d_v} \sum_{i=1}^{d_v} \braket{i|v_c}$ be the average of its amplitudes. Define the uniform state $\ket{v_{\sigma}}$ such that $\braket{i|v_{\sigma}} = \overline{v_c}$ for all $i$. Now consider the state $\ket{v_{\phi}} = \ket{v_c} - \ket{v_{\sigma}}$. It is a flip state, since
\begin{align*}
	\sum_{i=1}^{d_v} \braket{i|v_{\phi}} 
		&= \sum_{i=1}^{d_v} \braket{i|v} - \sum_{i=1}^{d_v} \braket{i|v_{\sigma}} \\
		&= \sum_{i=1}^{d_v} \braket{i|v} - d_v \cdot \overline{v_c} \\
		&= 0.
\end{align*}
Thus $\ket{v_c}$ can be expressed as a linear combination of uniform and flip states. For reference, let us write this as a Lemma:

\begin{lemma}
	\label{lemma:uniformflip}
	Any directional state $\ket{v_c}$ can be expressed as the sum of a uniform state $\ket{v_\sigma}$ and a flip state $\ket{v_\phi}$.
\end{lemma}

%-------------------------------------------------------------------------------

\subsection{General Stationary States}

Using uniform and flip states, we now derive if and only if conditions for a state to be stationary under the quantum walk search operator $U$ \eqref{eq:U}.

\begin{theorem} 
	\label{thm:stat-states}
	Let $\ket{\psi}$ be the state of the quantum walk. For each pair of adjacent vertices $a$ and $b$, let the amplitude on $\ket{ab}$ in $\ket{\psi}$ be $\sigma_1 + \phi_1$, where $\sigma_1$ comes from the uniform part of $a$ and $\phi_1$ comes from the flip part of $a$. Similarly, let the amplitude on $\ket{ba}$ in $\ket{\psi}$ be $\sigma_2 + \phi_2$. Then $\ket{\psi}$ is stationary under the quantum walk search operator $U = SCQ$ if and only if:
	\begin{enumerate}[(a)]
		\item	if $a$ is unmarked and $b$ is marked, then $\sigma_1 = \phi_2$ and $\phi_1 = -\sigma_2$;

		\item	if $a$ and $b$ are both unmarked, then $\sigma_1 = \sigma_2$ and $\phi_1 = -\phi_2$;

		\item	if $a$ and $b$ are both marked, then $\sigma_1 = -\sigma_2$ and $\phi_1 = \phi_2$.
	\end{enumerate}
\end{theorem}

\begin{figure}
\begin{center}
	\includegraphics{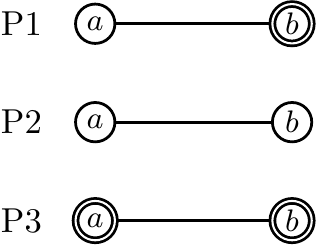}
	\caption{\label{fig:paircases} For a pair of adjacent vertices $a$ and $b$, the three possibilities of whether or not they are marked. The double circle indicates a marked vertex.}
\end{center}
\end{figure}

\begin{proof}
	For each pair of adjacent vertices $a$ and $b$, there are three possibilities for whether or not they are marked, as depicted in Fig.~\ref{fig:paircases}. Let us consider each of these possibilities.

	\emph{P1.} Suppose one vertex is unmarked and the other is marked. Without loss of generality, say $a$ is unmarked and $b$ marked. Now consider the action of $U = SCQ$. The oracle query flips the amplitude at marked vertices (in this case, the $b$ vertex), the coin inverts about the average (so it leaves stationary states alone and flips the sign of flip states), and the shift swaps the amplitudes at $\ket{ab}$ and $\ket{ba}$. Explicitly, here is what each operator does to the amplitudes:
	\begin{align*}
		\ket{ab}&: \sigma_1 + \phi_1 \xrightarrow{Q} \sigma_1 + \phi_1 \xrightarrow{C} \sigma_1 - \phi_1 \xrightarrow{S} -\sigma_2 + \phi_2 \\
		\ket{ba}&: \sigma_2 + \phi_2 \xrightarrow{Q} -\sigma_2 - \phi_2 \xrightarrow{C} -\sigma_2 + \phi_2 \xrightarrow{S} \sigma_1 - \phi_1.
	\end{align*}
	For $\ket{\psi}$ to be stationary, the amplitude before and after the application of $U$ must be the same:
	\begin{align*}
		\sigma_1 + \phi_1 &= -\sigma_2 + \phi_2 \\
		\sigma_2 + \phi_2 &= \sigma_1 - \phi_1.
	\end{align*}
	The solution is given by $\sigma_1 = \phi_2$ and $\phi_1 = -\sigma_2$.

	\emph{P2.} Now suppose that both $a$ and $b$ are unmarked. Then similarly, for $\ket{\psi}$ to be stationary, we have
	\begin{align*}
		\sigma_1 + \phi_1 &= \sigma_2 - \phi_2 \\
		\sigma_2 + \phi_2 &= \sigma_1 - \phi_1.
	\end{align*}
	The solution to this system is $\sigma_1 = \sigma_2$ and $\phi_1 = -\phi_2$.

	\emph{P3.} Lastly, suppose that both $a$ and $b$ are marked. Then for $\ket{\psi}$ to be stationary, we have
	\begin{align*}
		\sigma_1 + \phi_1 &= -\sigma_2 + \phi_2  \\
		\sigma_2 + \phi_2 &= -\sigma_1 + \phi_1.
	\end{align*}
	Here we have $\sigma_1 = -\sigma_2$ and $\phi_1 = \phi_2$.

	A stationary state satisfies all these properties, and a state that satisfies these properties for all edges is stationary.
\end{proof}

\begin{figure}
\begin{center}
	\includegraphics{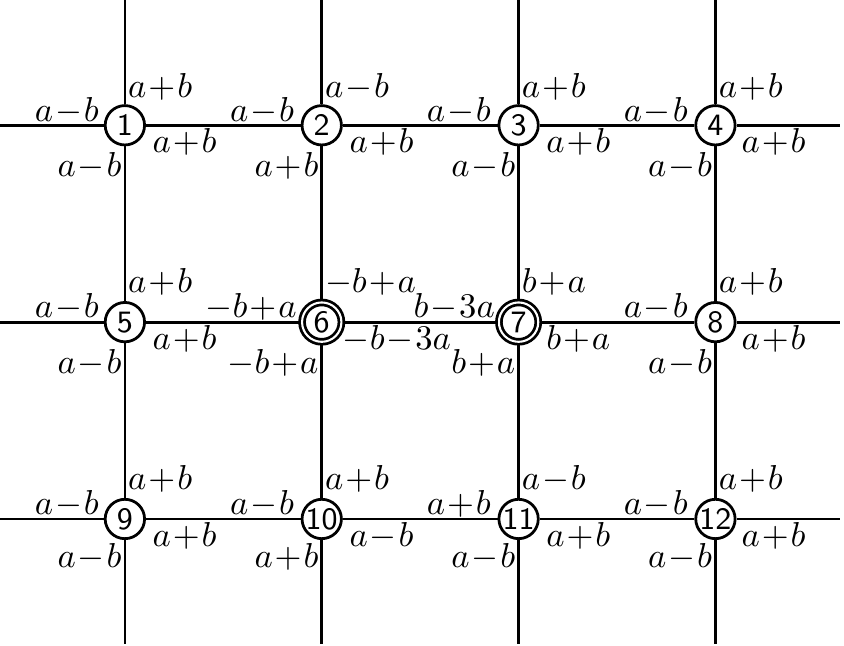}
	\caption{\label{fig:square_pair} A general stationary state for the $4 \times 3$ periodic square grid with an adjacent pair of marked vertices (vertices $6$ and $7$, indicated by double circles). For each amplitude, the uniform component is listed first followed by its flip component.}
\end{center}
\end{figure}

Figure~\ref{fig:square_pair} gives an example of a general stationary state for a pair of adjacent marked vertices on the 2D grid. For each edge, the amplitude is labeled by the uniform contribution plus the flip contribution. For example, consider the edge between vertices $5$ and $6$. For $\ket{5,6}$, we have amplitude $a + b$, which means $\sigma_1 = a$ and $\phi_1 = b$. For $\ket{6,5}$, we have amplitude $-b + a$, which means that $\sigma_2 = -b$ and $\phi_2 = a$. This satisfies P1, where $\sigma_1 = \phi_2 = a$ and $\phi_1 = -\sigma_2 = b$. Similarly, the rest of the edges satisfy Theorem~\ref{thm:stat-states}, so this is a stationary state.

Note that Theorem~\ref{thm:stat-states} exactly characterizes all stationary states of the search, and it greatly expands the number of known stationary states. For example, in Fig.~\ref{fig:square_pair}, $a$ and $b$ are continuous parameters (which also normalize the overall state), so there is an infinite number of stationary states. This contrasts with Nahimovs, Rivosh, and Santos \cite{NR2016,NS2016}, who only considered stationary states where unmarked vertices were uniform states and marked vertices were flip states, as in Fig.~\ref{fig:square_pair_opt}.

%-------------------------------------------------------------------------------

\subsection{Optimal Stationary States}

From the previous Theorem, a general stationary state can have both uniform and flip components at each vertex, and this can lead to an infinite number of stationary states. Algorithmically, however, we are interested in the \emph{optimal} stationary state, meaning the stationary state closest to the initial uniform state \eqref{eq:initial}. That is, we are interested in the stationary state $\ket{\psi}$ such that $\left|\braket{\psi(0) | \psi}\right|$ is maximized.

As we prove next, it turns out that this optimal stationary state is precisely the one described by Nahimovs, Rivosh, and Santos \cite{NR2016,NS2016}.

\begin{theorem} 
	\label{thm:opt-stat-states}
	The stationary state $\ket{\psi}$ maximizing $\left|\braket{\psi(0) | \psi}\right|$ satisfies the following properties:
	\begin{enumerate}[(a)]
		\item	the directional state of every unmarked vertex is a uniform state;
		\item	the directional state of every marked vertex is a flip state;
		\item	the amplitudes of adjacent vertices pointing to each other are equal. That is, $\braket{uv|\psi} = \braket{vu|\psi}$ for all $u \sim v$.
	\end{enumerate}
\end{theorem}

\begin{proof}
	We will prove that if $\ket{\psi}$ is an arbitrary stationary state, then the flip part of each unmarked vertex and the uniform part of each marked vertex together contribute zero to the inner product $\braket{\psi(0) | \psi}$. Hence we can remove them, resulting in each unmarked vertex being a uniform state and each marked vertex being a flip state. Then upon normalization, this has maximal overlap with the initial uniform state.

	The contribution from any flip state to the inner product $\braket{\psi(0) | \psi}$ is 0, since we showed in Lemma~\ref{lemma:uniformflip} that any flip state is orthogonal to a uniform state. Thus the flip parts of the unmarked vertices contribute nothing to $\braket{\psi(0) | \psi}$, so we remove them. This proves part (a) of the Theorem.

	Next we show that the total contribution from the uniform parts of the marked vertices to $\braket{\psi(0) | \psi}$ is equal to 0. Let the total sum of the amplitudes of $\ket{\psi}$ and $U\ket{\psi}$ be $s$ and $s'$, respectively. Beginning with $s$,
	\[ s = \sum_{v=1}^N \sum_{i=1}^{d_v} \braket{i | v_c}. \]
	From Lemma~\ref{lemma:uniformflip}, we can express a directional state $\ket{v_c} = \ket{v_{\sigma}} + \ket{v_{\phi}}$, where $\ket{v_{\sigma}}$ is a uniform state and $\ket{v_{\phi}}$ is a flip state. The flip states in the sum can be ignored, since $\sum_{i=1}^{d_v} \braket{i|v_\phi} = 0$. Then
	\[ s = \sum_{v=1}^N \sum_{i=1}^{d_v} \braket{i|v_{\sigma}}. \]
	Now for $s'$. The oracle flips the sign of the amplitude at marked vertices, so
	\[ s' = \sum_{v \notin M} \sum_{i=1}^{d_v} \braket{i|v_{\sigma}} - \sum_{v \in M} \sum_{i=1}^{d_v} \braket{i|v_{\sigma}}, \]
	where $M$ is the set of the marked vertices. Since $\ket{\psi}$ is stationary, $s = s'$. Therefore
	\[ \sum_{v \in M} \sum_{i=1}^{d_v} \braket{i|v_{\sigma}}=0. \]
	Thus if we look at the contribution from the uniform states at marked vertices to $\braket{\psi(0) | \psi}$, it is indeed equal to 0:
	\begin{align*}
		\sum_{v \in M} \sum_{i=1}^{d_v} \braket{v_c(0)|v_{\sigma}}
		&= \sum_{v \in M} \sum_{i=1}^{d_v} \braket{v_c(0)|i} \braket{i|v_{\sigma}} \\
		&= \sum_{v \in M} \sum_{i=1}^{d_v} \frac{1}{\sqrt{2|E|}} \braket{i|v_{\sigma}} \\
		&= \frac{1}{\sqrt{2|E|}} \sum_{v \in M} \sum_{i=1}^{d_v} \braket{i|v_{\sigma}} \\
		&= 0.
	\end{align*}
	Hence we remove the uniform parts from the marked vertices. This proves part (b) of the Theorem.
	
	Now that we have removed the flip components from unmarked vertices and the uniform components from marked vertices, consider what this did to the properties of Theorem~\ref{thm:stat-states}:
	\begin{itemize}
		\item	For P1, we now have $\phi_1 = \sigma_2 = 0$ and $\sigma_1 = \phi_2$.
		\item	For P2, we now have $\phi_1 = \phi_2 = 0$ and $\sigma_1 = \sigma_2$.
		\item	For P3, we now have $\sigma_1 = \sigma_2 = 0$ and $\phi_1 = \phi_2$.
	\end{itemize}
	So for the resulting state to be stationary (i.e., satisfy these three properties), we require that the amplitudes of adjacent vertices pointing to each other are equal, yielding part (c) of the Theorem.
\end{proof}

As an example of this reduction from a general stationary state to the optimal one, consider again the marked pair of adjacent vertices in Fig.~\ref{fig:square_pair}. We can remove the flip components from the unmarked vertices (the $\pm b$ parts) and the uniform components from the marked vertices (also the $\pm b$ parts), resulting in Fig.~\ref{fig:square_pair_opt}, which is the optimal stationary state (with normalization).

%-------------------------------------------------------------------------------
% Section
%-------------------------------------------------------------------------------

\section{Existence of Stationary States}

In Theorems~\ref{thm:stat-states} and \ref{thm:opt-stat-states}, general and optimal stationary states are given in terms of the relations between the uniform and flip components of each amplitude, depending on whether vertices are marked or not (cases P1, P2, and P3 in Fig.~\ref{fig:paircases}). In practice, however, finding a solution to all these conditions, if one exists, can be difficult. So in this section, we give two theorems for the existence of stationary states.

To investigate the existence of stationary states, it suffices to study optimal ones. This is because any general stationary state can be optimized to have maximal overlap with the initial uniform state \eqref{eq:initial} according to Theorem~\ref{thm:opt-stat-states}, namely by removing the flip components from unmarked vertices and the uniform components from unmarked vertices (and normalizing). So if no optimal stationary states exist for a certain graph and configuration of marked vertices, then no general stationary states exist, either. Furthermore, optimal stationary states determine how closely the initial uniform state \eqref{eq:initial} of the quantum walk search algorithm is to being stationary.

\begin{figure}
\begin{center}
	\includegraphics{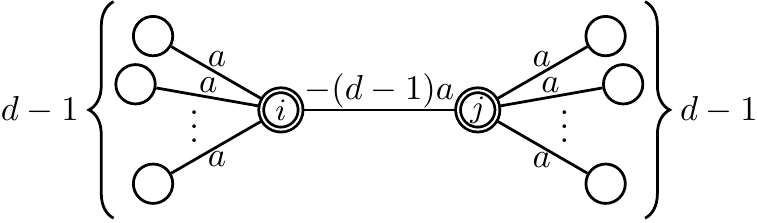}
	\caption{\label{fig:pair} The optimal stationary state for a marked pair of adjacent vertices of equal degree.}
\end{center}
\end{figure}

\begin{figure}
\begin{center}
	\includegraphics{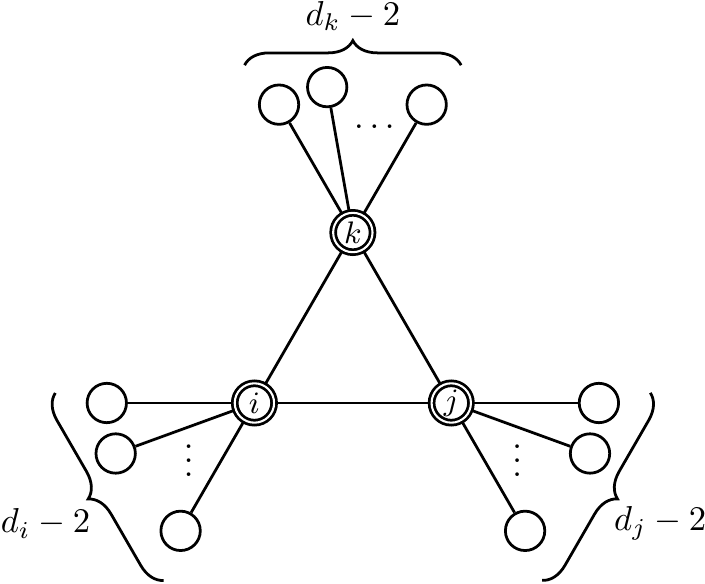}
	\caption{\label{fig:triangle} A marked triangle, where vertices can have unequal degree.}
\end{center}
\end{figure}

Before deriving our two theorems, let us provide some additional background. Nahimovs, Rivosh, and Santos \cite{NR2016,NS2016} showed that for general graphs, a marked pair of adjacent vertices has an optimal stationary state if both marked vertices have the same degree. This is depicted in Fig.~\ref{fig:pair}. Note that we only need to assign one amplitude to each edge of the graph since optimal stationary states have the same amplitude in both directions, i.e., part (c) of Theorem~\ref{thm:opt-stat-states}. This contrasts with a marked triangle in Fig.~\ref{fig:triangle}, where Nahimovs, Rivosh, and Santos gave the optimal stationary state, even if the marked vertices have different degrees. This raises the question of why equal degrees were used for the pair while unequal degrees were allowed for the triangle. Our next two theorems precisely explain why: It is because the pair is bipartite, which has constraints for the existence of stationary states, while the triangle is non-bipartite, which always has a stationary state.

Assume that the marked vertices form a connected component $M$. The unmarked vertices may form one or more connected components $U_1, \ldots, U_{k_u}$. To construct an optimal stationary state, we first assign the amplitudes in each $U_i$ so that each vertex is the same uniform state. This ensures that the unmarked vertices are uniform and the amplitudes of adjacent unmarked vertices pointing to each other are equal. Now we must determine how to assign the amplitudes of the marked vertices.

A marked vertex $v \in M$ is connected to some of the unmarked vertices. Since these unmarked vertices have already been assigned amplitudes, the amplitudes on these edges are also known. Let us say the sum of the amplitudes on these edges is equal to $\Sigma$. As $v$ is marked, it should be a flip state, so the sum of the edges going from $v$ to the other marked vertices is equal to $-\Sigma$. We call this value the \emph{shortage} at vertex $v$ and denote it by $s_v = -\Sigma$. Thus each vertex of $M$ has some shortage value, and we want to know when it is possible to assign the amplitudes on the edges between the marked vertices so as to neutralize all of the shortages.

In the following two theorems, we consider when $M$ is bipartite and non-bipartite. If it is bipartite, then the shortages can be neutralized if and only if the sum of the shortages on both partite sets are equal. On the other hand, if the marked vertices are non-bipartite, then the shortages can always be neutralized. This explains why Nahimovs, Rivosh, and Santos used equal degrees for the pair of marked vertices in Fig.~\ref{fig:pair}, whereas different degrees were allowed for the marked triangle in Fig.~\ref{fig:triangle}.

%-------------------------------------------------------------------------------

\subsection{\label{sec:bip} Bipartite Marked Component}

\begin{theorem}
	\label{thm:bip}
	If $M$ is bipartite, then we can assign amplitudes to neutralize the shortages at each marked vertex if and only if the sum of the shortages on both partite sets are equal.
\end{theorem}

\begin{figure}
\begin{center}
	\subfloat[]{
		\includegraphics{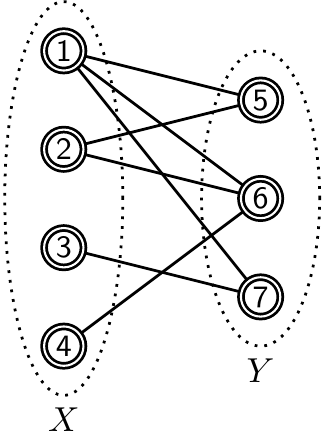}
		\label{fig:bipartite}
	} \quad
	\subfloat[]{
		\includegraphics{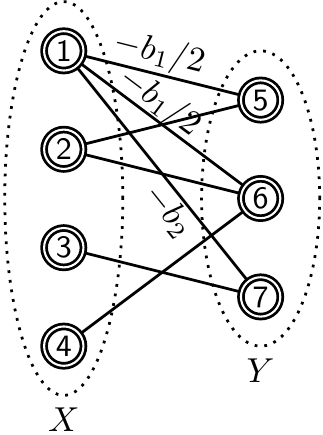}
		\label{fig:bipartite_assign}
	}

	\subfloat[]{
		\includegraphics{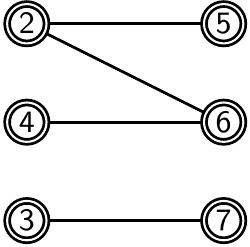}
		\label{fig:bipartite_disconnect}
	}
	\caption{(a) A bipartite marked connected component. (b) Assignments of vertex $1$'s edges within the marked connected component to neutralize its shortage, making it a flip state. (c) The marked vertices after removing vertex $1$.}
\end{center}
\end{figure}

\begin{proof}
	Let the partite sets be $X$ and $Y$. ($\Rightarrow$) If the state is optimally stationary, then the amplitudes of adjacent vertices pointing to each other are equal. Then the sum of the shortages at $X$ must be equal to that of $Y$.

	($\Leftarrow$) Now we assume that the sum of the shortages of both partite sets are equal. We prove that an optimal stationary state exists by giving a procedure for assigning the amplitudes.

	Pick any vertex $v \in M$ with $s_v \neq 0$. Without loss of generality suppose that $v \in X$. Suppose that removing $v$ and all its incident edges breaks up $M$ into connected components $C_1, \ldots, C_t$. For example, consider the marked connected component in Fig.~\ref{fig:bipartite}, which is bipartite. Consider vertex $v = 1$. If we remove it and its edges, then we have two connected components $C_1 = \{ 2, 4, 5, 6 \}$ and $C_2 = \{ 3, 7 \}$.

	For each connected component $C_i$, define $b_i$, where we take the shortages in the $X$ partite set and subtract the shortages in the $Y$ partite set:
	\[ b_i = \!\!\!\! \sum_{x \in C_i \cap X} \!\!\!\! s_x - \!\!\!\! \sum_{y \in C_i \cap Y} \!\!\!\! s_y. \]
	For our example, we have
	\begin{align*}
		b_1 &= s_2 + s_4 - s_5 - s_6 \\
		b_2 &= s_3 - s_7.
	\end{align*}
	Now consider the shortage at $v$, plus these $b_i$'s:
	\[ s_v + \sum_{i=1}^t b_i = \sum_{x \in X} s_x - \sum_{y \in Y} s_y = 0. \]
	This equals zero because the sum of the shortages of both partite sets are equal. Solving for the shortage at $v$, we find that it is equal to the negative of the sum of the $b_i$'s:
	\[ s_v = - \sum_{i=1}^t b_i. \]
	So for our example in Fig.~\ref{fig:bipartite}, we have $s_1 = -(b_1 + b_2)$.

	Let the number of edges from $v$ to $C_i$ be $n_i$. Then on each of these edges we assign an amplitude of $-b_i/n_i$. By construction, this neutralizes the shortage of $v$, since
	\[ \sum_{i=1}^t n_i \cdot\left(-\frac{b_i}{n_i}\right) = -\sum_{i=1}^t b_i = s_v. \]
	For our example, since $v$ connects to $C_1$ through two edges, we assign $-b_1/2$ to each of those edges. And since $v$ connects to $C_2$ through one edge, we assign $-b_2$ to that edge. This is shown in Fig.~\ref{fig:bipartite_assign}. By construction, we have now neutralized the shortage $s_1 = -(b_1 + b_2)$ at $v = 1$, making it a flip state.

	Besides neutralizing the shortage of $v$, this assignment also causes the sums of the shortages of $C_i$ on both of its partite sets to now be equal. For our example, in Fig.~\ref{fig:bipartite_assign}, the shortages of vertices 5, 6, and 7 have changed due to the assignments for $v = 1$. They are now
	\begin{align*}
		s_5' &= s_5 + \frac{b_1}{2} \\
		s_6' &= s_6 + \frac{b_1}{2} \\
		s_7' &= s_7 + b_2.
	\end{align*}
	Now let us sum the shortages for $C_1$. The sum in $X$ is $s_2 + s_4$, and the sum of the shortages in $Y$ is $s_5' + s_6' = s_5 + s_6 + b_1 = s_2 + s_4$. So they are equal. This similarly holds for $C_2$. In general, we subtract the assigned amplitudes from the shortages of the neighbors of $v$ and denote the new shortages by $s'$. Each neighbor of $v$ in $C_i$ is in $Y$, and each such neighbor has $s'_y=s_y+b_i/n_i$. As there are $n_i$ such neighbors, we have that
	\begin{align*}
		\sum_{x \in C_i \cap X} \!\!\!\! s'_x - \!\!\!\!\! \sum_{y \in C_i \cap Y} \!\!\!\! s'_y &= \!\!\!\! \sum_{x \in C_i \cap X} \!\!\!\! s_x - \!\!\!\!\! \sum_{y \in C_i \cap Y} \!\!\!\! s_y - b_i \\
		& = b_i-b_i = 0.
	\end{align*}

	Each $C_i$ is now a separate bipartite connected component, and the sum of the shortages on both partite sets in each $C_i$ are equal. Visualizing this for our example, removing vertex $v = 1$ leaves the two disconnected components, as shown in Fig.~\ref{fig:bipartite_disconnect}. Each of these $C_i$'s has the property that the sum of the shortages on both partite sets are equal, so we can recursively repeat the assignment procedure until the shortage at each vertex is $0$. The recursion stops when we have pairs, which can be made to have zero shortage. For example, consider the connected pair of vertices $3$ and $7$ in Fig.~\ref{fig:bipartite_disconnect}. Since the sum of shortages in $X$ and $Y$ are equal, we have that $s_3 = s_7$. Using the procedure, let $v = 3$. Then $C_1 = \{ 7 \}$ and $b_1 = -s_7$. So we assign the edge an amplitude of $-b_1 = s_7$. Since $s_3 = s_7$, we have neutralized the shortages of both vertices.
\end{proof}

Applying this to the marked pair of adjacent vertices in Fig.~\ref{fig:pair}, the unmarked vertices are assumed to form a single connected component, so all their amplitudes are $a$. Then the marked vertices have shortages $s_i = (d_i-1)a$ and $s_j = (d_j-1)a$. For a stationary state to exist, these shortages must be equal, which means the marked vertices must have equal degree $d_i = d_j$. This proves why Nahimovs, Rivosh, and Santos \cite{NR2016,NS2016} could only find a stationary state with this requirement.

\begin{figure}
\begin{center}
	\includegraphics{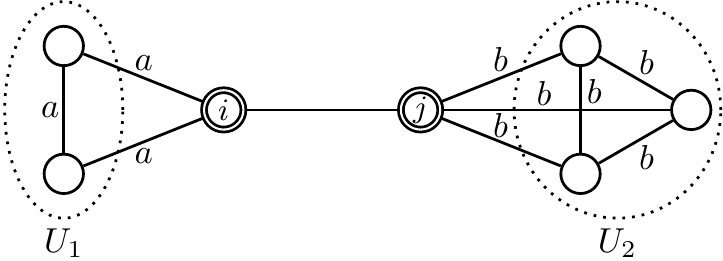}
	\caption{\label{fig:pair_disjoint} A graph with a marked pair of adjacent vertices and two unmarked connected components.}
\end{center}
\end{figure}

Now say the unmarked vertices form multiple connected components $U_1, \ldots, U_{k_u}$. For example, consider the marked pair of adjacent vertices in Fig.~\ref{fig:pair_disjoint}, where each marked vertex is connected to a different unmarked connected component $U_1$ and $U_2$. For each vertex in $U_1$, we assign all edges the same amplitude $a$, while for $U_2$, we assign the value $b$. Note that any value of $a$ and $b$ makes the unmarked vertices uniform, as desired. But for an optimal stationary state to exist, we specifically require that $2a = 3b$ so that the sum of the shortages on marked vertices $i$ and $j$ are equal.

In general, it may not be possible to assign uniform states to the $U_i$'s so that the sum of the shortages on the partite sets of $M$ are equal, so stationary states do not always exist. A simple example is a graph of two vertices connected by a single edge, where one of the two vertices is marked. Then there is no assignment to the edge that defines a stationary state. Furthermore, one can construct infinitely many examples where there is one unmarked component $U$ and one marked bipartite component $M$ such that the sums of the shortages cannot be equal no matter how the amplitudes are assigned on the unmarked component. Given this, we leave the details of how to assign the unmarked components for further research.

%-------------------------------------------------------------------------------

\subsection{Non-Bipartite Marked Component}

\begin{figure*}
\begin{center}
	\subfloat[]{
		\includegraphics{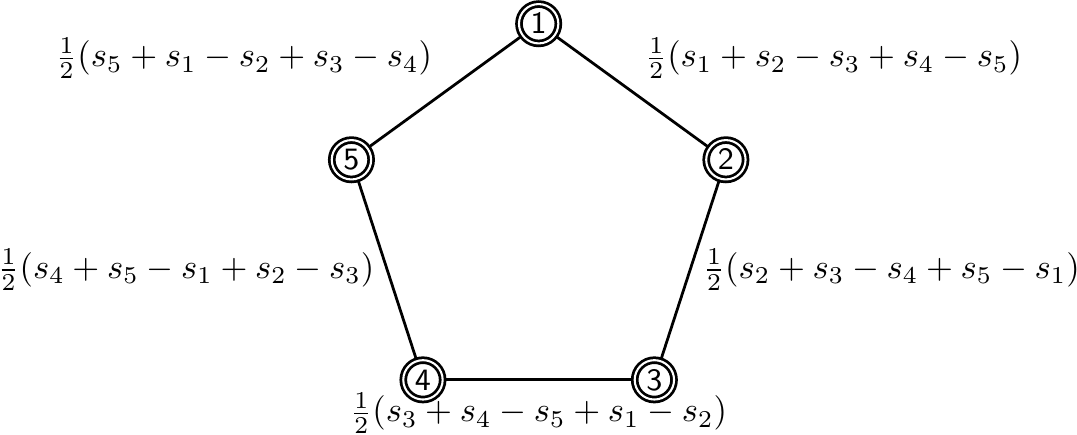}
		\label{fig:cycle}
	} \quad
	\subfloat[]{
		\includegraphics{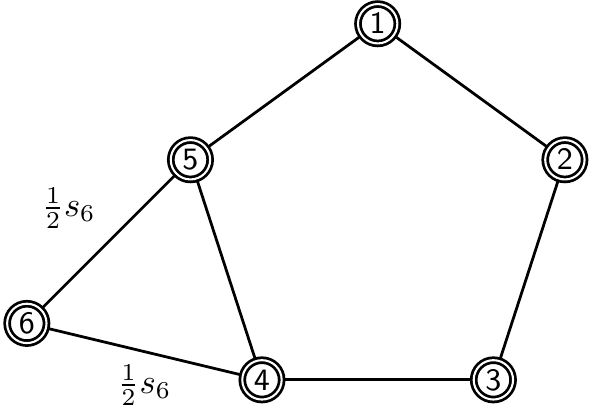}
		\label{fig:cycle_extra}
	}
	\caption{(a) A non-bipartite marked connected component, entirely constituting an odd cycle. Amplitudes have been assigned to the edges to neutralize the shortage at each vertex, yielding an optimal stationary state. (b) With an additional connected marked vertex. Amplitudes have been assigned to the edges to neutralize the shortage at the additional vertex.}
\end{center}
\end{figure*}

Now we consider the case when the marked vertices form a non-bipartite connected component $M$. In contrast to the previous theorem for bipartite graphs, here no constraints arise, so stationary states always exist for non-bipartite $M$.

\begin{theorem}
	\label{thm:non-bip}
	If $M$ is non-bipartite, then we can always assign the amplitudes to neutralize the shortages at each marked vertex.
\end{theorem}

\begin{proof}
	We prove the theorem by giving an explicit procedure for assigning the amplitudes. Since $M$ is non-bipartite, there exists a cycle of odd length $\{v_1, \ldots, v_k\}$.

	For now, suppose that this cycle contains all vertices of $M$. Then there always exists a solution for assigning the amplitudes on the edges of this cycle so to neutralize the shortages. It is given by
	\[ \braket{v_i v_{i+1}|\psi} = \braket{v_{i+1}v_i|\psi} = \frac 1 2 \sum_{j=1}^k (-1)^{(i-j) \Mod k} s_{v_j}. \]
	For example, for a 5-cycle of marked vertices, we assign the edges as shown in Fig.~\ref{fig:cycle}. So the assignment of the edges neutralizes the shortages. In general, we have
	\begin{align*}
		& \braket{v_{i-1}v_i|\psi} + \braket{v_i v_{i+1}|\psi} \\
		&\enspace = \frac 1 2 \sum_{j=1}^k \left[(-1)^{(i-1-j) \Mod k} + (-1)^{(i-j) \Mod k}\right] s_{v_j}.
	\end{align*}
	Here $x \Mod k$ is an integer from $0$ to $k-1$, so in particular, $(-1)^{-1 \Mod k} = (-1)^{k-1} = 1$ as $k$ is odd. The only time when $(i-1-j) \Mod k$ and $(i-j) \Mod k$ are equal$\Mod 2$ is when $i = j$. Otherwise they sum up to 0. Hence the value of this sum is equal to $s_{v_i}$.

	For all the other edges not in this cycle we assign amplitude 0. For example, in Fig.~\ref{fig:cycle}, if there were an edge connecting vertices $1$ and $3$, we would assign zero amplitude to it.

	Now suppose there are some vertices of $M$ not in this cycle. Pick a vertex $u$ from among them that is the farthest from the cycle. More formally, $\min_{i=1}^k d(u,v_i)$ is maximized for $u$, where $d(a,b)$ is the shortest distance between $a$ and $b$ in $M$. Let the degree of $u$ in this component be $\text{deg}(u)$. Assign an amplitude of $s_u/\text{deg}(u)$ for each edge from $u$ in this component. For example, Fig.~\ref{fig:cycle_extra} now includes an extra vertex labeled $6$, and we assign its edges within $M$ the value $s_6/2$, neutralizing its shortage.

	This way, the shortage of $u$ is neutralized and now we are left with a connected component $M \setminus u$, since $u$ was the farthest from the cycle. As $u$ did not belong to the cycle, we can repeat the procedure recursively until the cycle contains all vertices of the component.
\end{proof}

Again, this explains why Nahimovs, Rivosh, and Santos \cite{NR2016,NS2016} were able to find a stationary state for the marked triangle in Fig.~\ref{fig:triangle}, where each marked vertex could have different degrees.

%-------------------------------------------------------------------------------

\subsection{Multiple Marked Connected Components}

We end this section with some brief remarks about multiple marked connected components. In the previous two theorems, the marked vertices were all connected to each other, forming a single marked connected component $M$. Now say there are multiple such connected components $M_1, M_2, \ldots, M_{k_m}$. For the bipartite $M_i$'s, if the sum of the shortages of both partite sets are equal, then we can assign amplitudes to neutralize the shortages using Theorem~\ref{thm:bip}. As discussed in Section~\ref{sec:bip}, however, changing the uniform states of unmarked components $U_1, \ldots, U_{k_u}$, can change these sums. On the other hand, the non-bipartite $M_i$'s can always have their shortages neutralized using Theorem~\ref{thm:non-bip}.

%-------------------------------------------------------------------------------
% Section
%-------------------------------------------------------------------------------

\section{SKW Oracle}

Our results so far (Theorems \ref{thm:stat-states}--\ref{thm:non-bip}) pertained to the search operator $U = SCQ$ \eqref{eq:U}, which utilized the standard oracle $Q$ from Grover's algorithm that flips the sign of the marked vertices. We can rewrite this search operator another way, however:
\[ U = S \cdot \left\{ \begin{array}{rl}
	C & \text{on unmarked vertices} \\
	-C & \text{on marked vertices}
\end{array} \right\}. \]
This first applies the Grover diffusion coin $C$ on the unmarked vertices and the negative of it $-C$ to the marked vertices, which incorporates the oracle. Then it applies the flip-flop shift.

This motivates a different search operator, due to Shenvi, Kempe, and Whaley (SKW) \cite{SKW2003}, where we still apply $C$ to the unmarked vertices, but now apply $-I$ (minus the identity) to the marked vertices:
\[ U' = S \cdot \left\{ \begin{array}{rl}
	C & \text{on unmarked vertices} \\
	-I & \text{on marked vertices}
\end{array} \right\}. \]
Although this SKW operator $U'$ has a different notion of the oracle, it is equivalent to the usual Grover search operator $U$ whenever each marked vertex is only adjacent to identically-evolving vertices. For example, for search on the complete graph of $N$ vertices (i.e., the quantum walk formulation of Grover's algorithm) with a single marked vertex, the marked vertex is only adjacent to unmarked vertices, which evolve identically. Then the oracle/coin $-C$ in $U$ and $-I$ in $U'$ act identically on the marked vertex. This is explicitly proved in Section 2 of \cite{Wong10}, and after $\pi\sqrt{N} / 2\sqrt{2}$ applications of either operator, the system achieves a success probability of $1/2$ when measuring the position of the particle (although an internal-state measurement can further improve this \cite{Wong18}). As another example, the two operators $U$ and $U'$ are equivalent for search on arbitrary-dimensional periodic square lattices with a single marked vertex \cite{AKR2005}.

On the other hand, $U$ and $U'$ can behave very differently for other configurations. For example, from Section 6 of \cite{Wong10}, if there are multiple marked vertices $k \ge 2$ on the complete graph, then with $U$, the system reaches a success probability of $4k(k-1)/(2k-1)^2$ with a runtime of $\pi\sqrt{N} / \sqrt{2(2k-1)}$, whereas with $U'$, it reaches a success probability of $1/2$ (regardless of $k$) in time $\pi\sqrt{N} / 2\sqrt{2k}$. These two operators were also compared for the simplex of complete graphs with a fully marked clique in Fig.~\ref{fig:simplex_full} \cite{WA2015}, as well as the 2D periodic square lattice with multiple marked locations \cite{NR2015}.

Given the possibly distinct behavior of the SKW operator, we now investigate its stationary states, i.e., the 1-eigenvectors of $U'$. Perhaps surprisingly, we show in the following theorem that all such stationary states are orthogonal to the initial uniform state \eqref{eq:initial}.

\begin{theorem}
	\label{thm:SKW}
	With the SKW search operator $U'$, for every stationary state $\ket{\psi}$, we have $\braket{\psi(0)|\psi}=0$.
\end{theorem}

\begin{proof}
	As in Theorem~\ref{thm:stat-states}, we derive conditions on the amplitudes between adjacent vertices such that a state is stationary. As depicted in Fig.~\ref{fig:paircases}, there are three possibilities for how the adjacent vertices are marked.

	\emph{P1.} Suppose vertices $a$ and $b$ are adjacent, with $a$ unmarked and $b$ marked. Let the amplitude on $\ket{ab}$ be $\sigma_1 + \phi_1$, where $\sigma_1$ comes from the uniform part of $a$ and $\phi_1$ comes from the flip part of $a$. Similarly, let the amplitude on $\ket{ba}$ be $\sigma_2 + \phi_2$. Then after we apply $U'$, the amplitude on $\ket{ab}$ becomes $-\sigma_2 - \phi_2$, and the amplitude on $\ket{ba}$ becomes $\sigma_1 - \phi_1$. Since $\ket{\psi}$ is stationary, we have the equalities:
	\begin{align*}
		\sigma_1 + \phi_1 &= -\sigma_2 - \phi_2 \\
		\sigma_2 + \phi_2 &= \sigma_1 - \phi_1.
	\end{align*}
	The solution to this system is $\sigma_1 = 0$ and $\phi_1 = -\sigma_2-\phi_2$. 
	
	\emph{P2.} Now suppose that both $a$ and $b$ are unmarked. Then
	\begin{align*}
		\sigma_1 + \phi_1 &= \sigma_1 - \phi_2 \\
		\sigma_2 + \phi_2 &= \sigma_2 - \phi_1.
	\end{align*}
	Here we have $\sigma_1 = \sigma_2$ and $\phi_1 = -\phi_2$.

	\emph{P3.} Lastly, suppose that both $a$ and $b$ are marked. 
	\begin{align*}
		\sigma_1 + \phi_1 &= -\sigma_2 - \phi_2 \\
		\sigma_2 + \phi_2 &= -\sigma_1 - \phi_1.
	\end{align*}
	These two equations are equivalent.

	Note that in P1 and P3, we have $\sigma_1 + \phi_1 = -(\sigma_2 + \phi_2)$. We will now show that this is also true for P2. By P1, the uniform component at an unmarked vertex neighboring a marked one is 0. By P2, the uniform components of two neighboring unmarked vertices are equal. Therefore, assuming the graph is connected, the uniform component of every unmarked vertex is 0. Then in P2, we have $\sigma_1=\sigma_2=0$, and therefore $\sigma_1+\phi_1=-(\sigma_2+\phi_2)$.

	Thus for each edge pair of adjacent vertices $u$ and $v$, $\braket{uv |\psi} =-\braket{vu |\psi}$. But for the initial uniform state, $\braket{uv |\psi(0)} = \braket{vu |\psi(0)}$. Therefore each pair of adjacent vertices contributes 0 to $\braket{\psi(0)|\psi}$:
	\[ \braket{\psi(0) | uv}\braket{uv|\psi} + \braket{\psi(0) |vu }\braket{vu|\psi}=0. \]
	Hence summing over all the edges, $\braket{\psi(0)|\psi} = 0$.
\end{proof}

This theorem explains why, on the 2D periodic square grid, stationary states were found with $U$ \cite{NR2016,NS2016} but not with $U'$ \cite{AKR2005,NR2015,NR2016}. Similarly, the simplex of complete graphs with a fully marked clique, as in Fig.~\ref{fig:simplex_full}, has a stationary state for $U$ but not for $U'$ \cite{WA2015}. Finally, this proves that any configuration for which the Grover search operator $U$ and the SKW operator $U'$ are equivalent will not be disturbed by a stationary state.

%-------------------------------------------------------------------------------
% Section
%-------------------------------------------------------------------------------

\section{Conclusion}

Despite additional marked vertices making search by classical random walk easier, they can make search by quantum walk harder by causing the initial uniform state to be approximately equal to a stationary state, or 1-eigenvector, of the walk. We completely characterized such stationary states in terms of the amplitudes between adjacent vertices in Theorem~\ref{thm:stat-states}, and we showed in Theorem~\ref{thm:opt-stat-states} that the stationary states considered by Nahimovs, Rivosh, and Santos have maximal overlap with the initial state. We then investigated the existence of stationary states, showing in Theorem~\ref{thm:bip} that a bipartite marked connected component forms a stationary state if and only if the sums of the shortages on each bipartite set are equal. In contrast, a non-bipartite marked connected component always forms a stationary state, regardless of the shortages on each marked vertex, as shown in Theorem~\ref{thm:non-bip}. These results utilized the natural Grover oracle. Using the SKW oracle instead, we showed in Theorem~\ref{thm:SKW} that all stationary states are orthogonal to the initial uniform state. Hence in some cases, the SKW oracle may be superior to the more natural Grover oracle.

Although the 1-eigenvectors considered in our analysis are the only true stationary states of the quantum walk search operator, other eigenvectors are stationary in their probability distributions. Stationary probability distributions can also arise without the initial uniform state being an eigenvector, such as for a marked diagonal on the 2D periodic square lattice, where the walk alternates between $2N^2$ states without changing the probability at any vertex \cite{AR2008}. Such states are subjects for further research.

%-------------------------------------------------------------------------------
% Acknowledgments.
%-------------------------------------------------------------------------------

\begin{acknowledgments}
	Thanks to Nikolajs Nahimovs, Alexander Rivosh, and Raqueline A.~M.~Santos for sharing their work on stationary states with us.
	This work was supported by the European Union Seventh Framework Programme (FP7/2007-2013) under the QALGO (Grant Agreement No.~600700) project and the RAQUEL (Grant Agreement No.~323970) project, the ERC Advanced Grant MQC, and the Latvian State Research Programme NeXIT project No.~1.
\end{acknowledgments}

%-------------------------------------------------------------------------------
% References.
%-------------------------------------------------------------------------------

\bibliography{refs}

\begin{thebibliography}{16}%
\makeatletter
\providecommand \@ifxundefined [1]{%
 \@ifx{#1\undefined}
}%
\providecommand \@ifnum [1]{%
 \ifnum #1\expandafter \@firstoftwo
 \else \expandafter \@secondoftwo
 \fi
}%
\providecommand \@ifx [1]{%
 \ifx #1\expandafter \@firstoftwo
 \else \expandafter \@secondoftwo
 \fi
}%
\providecommand \natexlab [1]{#1}%
\providecommand \enquote  [1]{``#1''}%
\providecommand \bibnamefont  [1]{#1}%
\providecommand \bibfnamefont [1]{#1}%
\providecommand \citenamefont [1]{#1}%
\providecommand \href@noop [0]{\@secondoftwo}%
\providecommand \href [0]{\begingroup \@sanitize@url \@href}%
\providecommand \@href[1]{\@@startlink{#1}\@@href}%
\providecommand \@@href[1]{\endgroup#1\@@endlink}%
\providecommand \@sanitize@url [0]{\catcode `\\12\catcode `\$12\catcode
  `\&12\catcode `\#12\catcode `\^12\catcode `\_12\catcode `\%12\relax}%
\providecommand \@@startlink[1]{}%
\providecommand \@@endlink[0]{}%
\providecommand \url  [0]{\begingroup\@sanitize@url \@url }%
\providecommand \@url [1]{\endgroup\@href {#1}{\urlprefix }}%
\providecommand \urlprefix  [0]{URL }%
\providecommand \Eprint [0]{\href }%
\providecommand \doibase [0]{http://dx.doi.org/}%
\providecommand \selectlanguage [0]{\@gobble}%
\providecommand \bibinfo  [0]{\@secondoftwo}%
\providecommand \bibfield  [0]{\@secondoftwo}%
\providecommand \translation [1]{[#1]}%
\providecommand \BibitemOpen [0]{}%
\providecommand \bibitemStop [0]{}%
\providecommand \bibitemNoStop [0]{.\EOS\space}%
\providecommand \EOS [0]{\spacefactor3000\relax}%
\providecommand \BibitemShut  [1]{\csname bibitem#1\endcsname}%
\let\auto@bib@innerbib\@empty
%</preamble>
\bibitem [{\citenamefont {Montanaro}(2016)}]{Montanaro2016}%
  \BibitemOpen
  \bibfield  {author} {\bibinfo {author} {\bibfnamefont {Ashley}\ \bibnamefont
  {Montanaro}},\ }\bibfield  {title} {\enquote {\bibinfo {title} {Quantum
  algorithms: an overview},}\ }\href {\doibase 10.1038/npjqi.2015.23}
  {\bibfield  {journal} {\bibinfo  {journal} {NPJ Quantum Inf.}\ }\textbf
  {\bibinfo {volume} {2}},\ \bibinfo {pages} {15023} (\bibinfo {year}
  {2016})}\BibitemShut {NoStop}%
\bibitem [{\citenamefont {Grover}(1996)}]{Grover1996}%
  \BibitemOpen
  \bibfield  {author} {\bibinfo {author} {\bibfnamefont {L.~K.}\ \bibnamefont
  {Grover}},\ }\bibfield  {title} {\enquote {\bibinfo {title} {A fast quantum
  mechanical algorithm for database search},}\ }in\ \href@noop {} {\emph
  {\bibinfo {booktitle} {Proceedings of the 28th Annual ACM Symposium on Theory
  of Computing}}},\ \bibinfo {series and number} {STOC '96}\ (\bibinfo
  {publisher} {ACM},\ \bibinfo {address} {New York, NY, USA},\ \bibinfo {year}
  {1996})\ pp.\ \bibinfo {pages} {212--219}\BibitemShut {NoStop}%
\bibitem [{\citenamefont {Kempe}(2003)}]{Kempe2003}%
  \BibitemOpen
  \bibfield  {author} {\bibinfo {author} {\bibfnamefont {Julia}\ \bibnamefont
  {Kempe}},\ }\bibfield  {title} {\enquote {\bibinfo {title} {Quantum random
  walks: An introductory overview},}\ }\href@noop {} {\bibfield  {journal}
  {\bibinfo  {journal} {Contemp. Phys.}\ }\textbf {\bibinfo {volume} {44}},\
  \bibinfo {pages} {307--327} (\bibinfo {year} {2003})}\BibitemShut {NoStop}%
\bibitem [{\citenamefont {Meyer}(1996{\natexlab{a}})}]{Meyer1996a}%
  \BibitemOpen
  \bibfield  {author} {\bibinfo {author} {\bibfnamefont {David~A.}\
  \bibnamefont {Meyer}},\ }\bibfield  {title} {\enquote {\bibinfo {title} {From
  quantum cellular automata to quantum lattice gases},}\ }\href {\doibase
  10.1007/BF02199356} {\bibfield  {journal} {\bibinfo  {journal} {J. Stat.
  Phys.}\ }\textbf {\bibinfo {volume} {85}},\ \bibinfo {pages} {551--574}
  (\bibinfo {year} {1996}{\natexlab{a}})}\BibitemShut {NoStop}%
\bibitem [{\citenamefont {Meyer}(1996{\natexlab{b}})}]{Meyer1996b}%
  \BibitemOpen
  \bibfield  {author} {\bibinfo {author} {\bibfnamefont {David~A.}\
  \bibnamefont {Meyer}},\ }\bibfield  {title} {\enquote {\bibinfo {title} {On
  the absence of homogeneous scalar unitary cellular automata},}\ }\href
  {\doibase http://dx.doi.org/10.1016/S0375-9601(96)00745-1} {\bibfield
  {journal} {\bibinfo  {journal} {Phys. Lett. A}\ }\textbf {\bibinfo {volume}
  {223}},\ \bibinfo {pages} {337 -- 340} (\bibinfo {year}
  {1996}{\natexlab{b}})}\BibitemShut {NoStop}%
\bibitem [{\citenamefont {Ambainis}\ \emph {et~al.}(2005)\citenamefont
  {Ambainis}, \citenamefont {Kempe},\ and\ \citenamefont {Rivosh}}]{AKR2005}%
  \BibitemOpen
  \bibfield  {author} {\bibinfo {author} {\bibfnamefont {Andris}\ \bibnamefont
  {Ambainis}}, \bibinfo {author} {\bibfnamefont {Julia}\ \bibnamefont {Kempe}},
  \ and\ \bibinfo {author} {\bibfnamefont {Alexander}\ \bibnamefont {Rivosh}},\
  }\bibfield  {title} {\enquote {\bibinfo {title} {Coins make quantum walks
  faster},}\ }in\ \href {http://dl.acm.org/citation.cfm?id=1070432.1070590}
  {\emph {\bibinfo {booktitle} {Proceedings of the 16th Annual ACM-SIAM
  Symposium on Discrete Algorithms}}},\ \bibinfo {series and number} {SODA
  '05}\ (\bibinfo  {publisher} {SIAM},\ \bibinfo {address} {Philadelphia, PA,
  USA},\ \bibinfo {year} {2005})\ pp.\ \bibinfo {pages}
  {1099--1108}\BibitemShut {NoStop}%
\bibitem [{\citenamefont {Shenvi}\ \emph {et~al.}(2003)\citenamefont {Shenvi},
  \citenamefont {Kempe},\ and\ \citenamefont {Whaley}}]{SKW2003}%
  \BibitemOpen
  \bibfield  {author} {\bibinfo {author} {\bibfnamefont {Neil}\ \bibnamefont
  {Shenvi}}, \bibinfo {author} {\bibfnamefont {Julia}\ \bibnamefont {Kempe}}, \
  and\ \bibinfo {author} {\bibfnamefont {K.~Birgitta}\ \bibnamefont {Whaley}},\
  }\bibfield  {title} {\enquote {\bibinfo {title} {Quantum random-walk search
  algorithm},}\ }\href {\doibase 10.1103/PhysRevA.67.052307} {\bibfield
  {journal} {\bibinfo  {journal} {Phys. Rev. A}\ }\textbf {\bibinfo {volume}
  {67}},\ \bibinfo {pages} {052307} (\bibinfo {year} {2003})}\BibitemShut
  {NoStop}%
\bibitem [{\citenamefont {Ambainis}\ \emph {et~al.}(2016)\citenamefont
  {Ambainis}, \citenamefont {Pr\={u}sis}, \citenamefont {Vihrovs},\ and\
  \citenamefont {Wong}}]{APVW2016}%
  \BibitemOpen
  \bibfield  {author} {\bibinfo {author} {\bibfnamefont {Andris}\ \bibnamefont
  {Ambainis}}, \bibinfo {author} {\bibfnamefont {Kri\v{s}j\={a}nis}\
  \bibnamefont {Pr\={u}sis}}, \bibinfo {author} {\bibfnamefont {Jevg\={e}nijs}\
  \bibnamefont {Vihrovs}}, \ and\ \bibinfo {author} {\bibfnamefont {Thomas~G.}\
  \bibnamefont {Wong}},\ }\bibfield  {title} {\enquote {\bibinfo {title}
  {Oscillatory localization of quantum walks by classical electric circuits},}\
  }\href@noop {} {\bibfield  {journal} {\bibinfo  {journal} {{a}rXiv:1606.02136
  [quant-ph]}\ } (\bibinfo {year} {2016})}\BibitemShut {NoStop}%
\bibitem [{\citenamefont {Brassard}\ \emph {et~al.}(2002)\citenamefont
  {Brassard}, \citenamefont {H\o{}yer}, \citenamefont {Mosca},\ and\
  \citenamefont {Tapp}}]{BHMT2000}%
  \BibitemOpen
  \bibfield  {author} {\bibinfo {author} {\bibfnamefont {Gilles}\ \bibnamefont
  {Brassard}}, \bibinfo {author} {\bibfnamefont {Peter}\ \bibnamefont
  {H\o{}yer}}, \bibinfo {author} {\bibfnamefont {Michele}\ \bibnamefont
  {Mosca}}, \ and\ \bibinfo {author} {\bibfnamefont {Alain}\ \bibnamefont
  {Tapp}},\ }\enquote {\bibinfo {title} {Quantum amplitude amplification and
  estimation},}\ in\ \href {\doibase 10.1090/conm/305} {\emph {\bibinfo
  {booktitle} {Quantum computation and information}}},\ \bibinfo {series}
  {Contemp. Math.}, Vol.\ \bibinfo {volume} {305}\ (\bibinfo  {publisher}
  {Amer. Math. Soc.},\ \bibinfo {address} {Providence, RI},\ \bibinfo {year}
  {2002})\ pp.\ \bibinfo {pages} {53--–74}\BibitemShut {NoStop}%
\bibitem [{\citenamefont {Nahimovs}\ and\ \citenamefont
  {Rivosh}(2016)}]{NR2016}%
  \BibitemOpen
  \bibfield  {author} {\bibinfo {author} {\bibfnamefont {Nikolajs}\
  \bibnamefont {Nahimovs}}\ and\ \bibinfo {author} {\bibfnamefont {Alexander}\
  \bibnamefont {Rivosh}},\ }\bibfield  {title} {\enquote {\bibinfo {title}
  {Exceptional configurations of quantum walks with {G}rover's coin},}\ }in\
  \href {\doibase 10.1007/978-3-319-29817-7_8} {\emph {\bibinfo {booktitle}
  {Proceedings of the 10th International Doctoral Workshop on Mathematical and
  Engineering Methods in Computer Science}}},\ \bibinfo {series and number}
  {MEMICS 2015}\ (\bibinfo  {publisher} {Springer},\ \bibinfo {address}
  {Tel{\v{c}}, Czech Republic},\ \bibinfo {year} {2016})\ pp.\ \bibinfo {pages}
  {79--92}\BibitemShut {NoStop}%
\bibitem [{\citenamefont {Nahimovs}\ and\ \citenamefont
  {Santos}(2016)}]{NS2016}%
  \BibitemOpen
  \bibfield  {author} {\bibinfo {author} {\bibfnamefont {Nikolajs}\
  \bibnamefont {Nahimovs}}\ and\ \bibinfo {author} {\bibfnamefont {Raqueline
  A.~M.}\ \bibnamefont {Santos}},\ }\bibfield  {title} {\enquote {\bibinfo
  {title} {Adjacent vertices can be hard to find by quantum walks},}\
  }\href@noop {} {\bibfield  {journal} {\bibinfo  {journal} {{a}rXiv:1605.05598
  [quant-ph]}\ } (\bibinfo {year} {2016})}\BibitemShut {NoStop}%
\bibitem [{\citenamefont {Wong}\ and\ \citenamefont {Ambainis}(2015)}]{WA2015}%
  \BibitemOpen
  \bibfield  {author} {\bibinfo {author} {\bibfnamefont {Thomas~G.}\
  \bibnamefont {Wong}}\ and\ \bibinfo {author} {\bibfnamefont {Andris}\
  \bibnamefont {Ambainis}},\ }\bibfield  {title} {\enquote {\bibinfo {title}
  {Quantum search with multiple walk steps per oracle query},}\ }\href
  {\doibase 10.1103/PhysRevA.92.022338} {\bibfield  {journal} {\bibinfo
  {journal} {Phys. Rev. A}\ }\textbf {\bibinfo {volume} {92}},\ \bibinfo
  {pages} {022338} (\bibinfo {year} {2015})}\BibitemShut {NoStop}%
\bibitem [{\citenamefont {Wong}(2015)}]{Wong10}%
  \BibitemOpen
  \bibfield  {author} {\bibinfo {author} {\bibfnamefont {Thomas~G.}\
  \bibnamefont {Wong}},\ }\bibfield  {title} {\enquote {\bibinfo {title}
  {Grover search with lackadaisical quantum walks},}\ }\href {\doibase
  10.1088/1751-8113/48/43/435304} {\bibfield  {journal} {\bibinfo  {journal}
  {J. Phys. A: Math. Theor.}\ }\textbf {\bibinfo {volume} {48}},\ \bibinfo
  {pages} {435304} (\bibinfo {year} {2015})}\BibitemShut {NoStop}%
\bibitem [{\citenamefont {Pr\={u}sis}\ \emph {et~al.}(2015)\citenamefont
  {Pr\={u}sis}, \citenamefont {Vihrovs},\ and\ \citenamefont {Wong}}]{Wong18}%
  \BibitemOpen
  \bibfield  {author} {\bibinfo {author} {\bibfnamefont {Kri\v{s}j\={a}nis}\
  \bibnamefont {Pr\={u}sis}}, \bibinfo {author} {\bibfnamefont {Jevg\={e}nijs}\
  \bibnamefont {Vihrovs}}, \ and\ \bibinfo {author} {\bibfnamefont {Thomas~G.}\
  \bibnamefont {Wong}},\ }\bibfield  {title} {\enquote {\bibinfo {title}
  {Doubling the success of quantum walk search using internal-state
  measurements},}\ }\href@noop {} {\bibfield  {journal} {\bibinfo  {journal}
  {{a}rXiv:1511.03865 [quant-ph]}\ } (\bibinfo {year} {2015})}\BibitemShut
  {NoStop}%
\bibitem [{\citenamefont {Nahimovs}\ and\ \citenamefont
  {Rivosh}(2015)}]{NR2015}%
  \BibitemOpen
  \bibfield  {author} {\bibinfo {author} {\bibfnamefont {Nikolajs}\
  \bibnamefont {Nahimovs}}\ and\ \bibinfo {author} {\bibfnamefont {Alexander}\
  \bibnamefont {Rivosh}},\ }\bibfield  {title} {\enquote {\bibinfo {title}
  {Quantum walks on two-dimensional grids with multiple marked locations},}\
  }\href@noop {} {\bibfield  {journal} {\bibinfo  {journal} {{a}rXiv:1507.03788
  [quant-ph]}\ } (\bibinfo {year} {2015})}\BibitemShut {NoStop}%
\bibitem [{\citenamefont {Ambainis}\ and\ \citenamefont
  {Rivosh}(2008)}]{AR2008}%
  \BibitemOpen
  \bibfield  {author} {\bibinfo {author} {\bibfnamefont {Andris}\ \bibnamefont
  {Ambainis}}\ and\ \bibinfo {author} {\bibfnamefont {Alexander}\ \bibnamefont
  {Rivosh}},\ }\bibfield  {title} {\enquote {\bibinfo {title} {Quantum walks
  with multiple or moving marked locations},}\ }in\ \href {\doibase
  10.1007/978-3-540-77566-9_42} {\emph {\bibinfo {booktitle} {Proceedings of
  the 34th Conference on Current Trends in Theory and Practice of Computer
  Science}}},\ \bibinfo {series and number} {SOFSEM 2008}\ (\bibinfo
  {publisher} {Springer},\ \bibinfo {address} {Nov\'y Smokovec, Slovakia},\
  \bibinfo {year} {2008})\ pp.\ \bibinfo {pages} {485--496}\BibitemShut
  {NoStop}%
\end{thebibliography}%

\end{document}